\documentclass[reqno,keywordsasfootnote]{article}
\usepackage[english]{babel}
\usepackage[latin1]{inputenc}
\usepackage{amsmath,amsthm,amsfonts,amssymb,times}
\numberwithin{equation}{section}
\usepackage[final]{epsfig}

\usepackage{times}

\newtheorem{theorem}{Theorem}[section]

\newtheorem{lemma}[theorem]{Lemma}

\theoremstyle{definition}
\newtheorem{definition}{Definition}[section]

\newcommand{\E}{\mathbb{E}}
\newcommand{\Z}{ {\mathbb Z} }
\newcommand{\N} {{\mathbb N}}

\newcommand{\R}{\mathbb{R}}

\DeclareMathOperator{\indfct}{\rm 1}

\DeclareMathOperator{\dist}{dist}

\DeclareMathOperator{\diam}{diam}

\newcommand{\bx}{\mathbf{x}}
\newcommand{\by}{\mathbf{y}}

\newcommand{\balpha}{\boldsymbol{\alpha}}

\date{September 29, 2009}

\begin{document}

\title{\LARGE Complete Dynamical Localization in Disordered\\  
Quantum Multi-Particle Systems\footnote{
Presented by S.~Warzel at the XVI International Congress of Mathematical Physics,  Prague 2009.} 
}



\author{   \Large Michael Aizenman$^{(a)}$   \   and   \  \  
 Simone Warzel$^{(b)}$
  \\[2ex]    
{ $^{(a)}$ Departments of Physics and Mathematics }  \\
{ Princeton University,  Princeton NJ 08544, USA  } \\[2ex]
{ $^ {(b)}$  Zentrum Mathematik, TU M\"unchen}\\ 
{   Boltzmannstr. 3, 85747 Garching, Germany} 
\\[1ex]     
  }

\maketitle

\begin{abstract}
We present some recent results concerning the persistence of dynamical localization for disordered systems of $ n $ particles under weak 
interactions.
\end{abstract}


\section{Introduction}\label{aba:sec1}

After more than half a century, Anderson localization continues to attract the interest of a broad spectrum of researchers ranging from experimentalists, 
who finds its effects in  systems of cold atoms and in photonic crystals \cite{AI09}, to mathematical physicists.  
Considerable progress was made in the rigorous methods for  the study of the localization effects of disorder in the context of the one-particle theory~\cite{St00,Kis,AiWa_RSO}.  
 More recently attention has centered on the role of interactions, and questions related to the persistence of the localization picture in the presence of inter-particle interactions.   We shall report here on some progress which was made in that area.


\subsection{The one-particle theory} 

The discussion of weakly interacting particles often starts  from the approximation in which the interactions are ignored.  That is, one first considers systems of Fermions, or Bosons possibly with on-site repulsion, subject to a fixed potential which includes random terms.   Such system can be understood in terms of the one-particle theory.  


For particles moving on a lattice $ \mathbb{Z}^d $, the single-particle Hamiltonian may take the form
\begin{equation}\label{eq:defH}
	H^{(1)}(\omega) := - \Delta \, + \, \lambda \, V(x;\omega) \, , \qquad \mbox{in $ \ell^2(\mathbb{Z}^d) $,} 
\end{equation}
where $-\Delta$ is the lattice Laplacian 
and the random potential incorporates a disorder parameter $ \lambda \geq  0 $. 
For the convenience of presentation, all mathematical statements in this note are made under the assumption that the random potential  $ V(x;\omega) $  takes  
 independent and 
indentically distributed values on the lattice sites $ x \in \mathbb{Z}^d $, with a distribution
 $\,  \mathbb{P}(V_x \in dv ) = \varrho(v) \, dv \, $ of a bounded, compactly supported density.
 

When talking about localization for operators such as $ H^{(1)}(\omega) $, different notions have been used and established:
\begin{description}
\item[Spectral localization] in an energy regime $ I \subset \mathbb{R}$, refers to the statement that within the specified energy regime $  H^{(1)}(\omega) $   has only pure point spectrum, 
with exponentially localized eigenfunctions, cf.~\cite{St00} and references therein.
\item[Dynamical localization] in $ I \subset \mathbb{R}$,   refers to the statement that any initially localized state, which  is a wave packet with energies in $ I \subset \mathbb{R}$, will 
 remain exponentially localized indefinitely under the time evolution.   A convenient sufficient condition is that 
for some $ A,  \xi \in (0, \infty) $ (both depending on $ I $, $ \lambda $), 
\begin{equation*}\label{eq:dynloc}
\mbox{(DL)}\qquad\qquad  \mathbb{E}\left[\sup_{t\in \mathbb{R}} \big| \langle \delta_y \, , e^{-itH^{(1)}} P_I(H^{(1)})\, \delta_x \rangle \big|^2\right] \leq A \, e^{-|x-y|/\xi} 
	\qquad\qquad 
\end{equation*}
where $ \delta_x \in \ell^2(\mathbb{Z}^d) $  is the  $\delta$ state at $ x  $ and $ P_I(H^{(1)})  $ is a  spectral projection. 
\end{description}
While spectral localization provides some coarse information on the quantum time evolution via the RAGE theorem, dynamical localization in the sense of~(DL) is a somewhat stronger notion.   
Historically,  spectral localization was the first to be rigorously established by building, for $d>1$ at  extreme energies or large disorder, 
on the multiscale analysis of Fr\"ohlich and Spencer~\cite{FrSp83}; cf.~\cite{St00,Kis}.  
The first proof of dynamical localization \cite{A2} relied on the fractional moment analysis of \cite{AM}. 

It is one of the features of the single particle model \eqref{eq:defH} that in any dimension there is an extreme disorder regime, with $\lambda $ exceeding some finite $ \lambda_1$,  in which  one has 
{\bf complete dynamical localization}, meaning that {\it all} states are localized and  
dynamical localization~(DL) holds with $ I = \mathbb{R} $.
%
%

For the single particle model~\eqref{eq:defH} in $ d = 1 $ complete dynamical localization with critical disorder strength $ \lambda_1 = 0 $ 
was established  already in~\cite{KuSou80}. For higher dimensions, $ d \geq 1 $, an explicit bound on $ \lambda_1 $ can be found in \cite{A2}.  

\subsection{The influence of interaction on localization: a challenge} 
It is not hard to see that the above localization properties are inherited by a system of non-interacting particles with a one-particle Hamiltonian corresponding to~\eqref{eq:defH}.  
An important question however, is whether such behavior will persist under the addition of interparticle interactions.  
Especially interesting is the situation where there are $n$ fermions in a region of volume $|\Lambda|$, with 
$|\Lambda| \to \infty$ and  $n/|\Lambda| \to \rho >0$.  

A rather strong claim is being advocated by Basko, Aleiner and Altshuler~\cite{BAA06}, who argue that \emph{under strong disorder weak interactions do not change the qualitative picture of localization  as it is seen in the non-interacting model}.    

The proposal is rather startling.  It includes the claim that if the system of particles of an overall positive density is started at a initial state at which the distribution in space of particles and energy is far from uniform, its irregularity will persist indefinitely under the time evolution.   This seems to run against the vague \emph{equidistribution principle}, by which one expects that 
except under unusual circumstances, such as in the non-interacting integrable model, the initial state will evolve in time (in a weak enough sense) towards states which maximize the entropy, in a coarse-grained sense, subject to the given energy and particle number constraints.  

%

%

Persistence of localization is not undisputed among physicist. In particular,    
long-range repulsive interactions are conjectured to have a delocalizing effect~\cite{Shep}. 

Rigorous methods  are still  far from allowing one to decide whether complete localization will persist or perish in the presence of interactions.  
Furthermore,  the analysis of even  a fixed number of particles with short range interactions in the infinite-volume limit has presented difficulties, and we will
report on some recent progress \cite{ChSu3,CS09,AW09} made in this direction. 


As an side we note  that  the  dynamics of multiparticle systems bear some relation   to non-linear evolutions.  Recent results on that topic  are discussed in~\cite{WZ}. 

\section{Dynamical localization for multi-particle systems} 
\subsection{The $ n $-particle Hamiltonian}
We will be concerned with a system of finitely  many interacting particles in the random potential described above. 
The $ n $-particle Hamiltonian is given by 
\begin{equation}\label{eq:defHm}
H^{(n)}(\omega) := \sum_{j=1}^n \left[- \Delta_j \, + \, \lambda \, V(x_j;\omega) \right]\, + {\mathcal U}(\bx;\balpha) , \qquad 
\mbox{in \quad $ \ell^2(\mathbb{Z}^d)^n $} 
\end{equation}
acting in the Hilbert space over all configurations $ \bx =(x_1, \dots, x_n) \in (\mathbb{Z}^d)^n $. 
Here the last term is a $ p $-site interaction of range $ \ell_U < \infty $:
$${\mathcal U}(\bx;\balpha)  \ :=\ \sum_{k=1}^{p}  \alpha_k  \mkern-10mu \sum_{
\substack{ 
A\subset \Z^d:  |A|=k \\ 
{\rm diam}A \le \ell_U
}  }  \mkern-5mu   U_A( (N_{u}(\bx))_{u\in A})  \, . 
$$
It is described in term of an interaction parameter $ \balpha = (\alpha_1,\dots,\alpha_p) \in \R^p$ and a function
 $ U_A : \mathbb{N}^{|A|} \to \R $ which is bounded by some $n$-dependent constant, $ \| U_A \| \leq c_n < \infty$,  and translational invariant. The latter 
depends on the number $N_{u}(\bx) := \sum_{j=1}^n \delta_{u,x_j} $ of particles of the configuration $ \bx =(x_1, \dots, x_n) $, which are at sites $ u \in A$ 
in the pattern $A \subset \Z^d $. 

Simply stated examples which are already of interest  are short range pair interactions for which 
$ p = 2 $  and 
$ U_{\{u,v\}}\left(  N_{u}(\bx), N_{v}(\bx) \right) =    N_{u}(\bx)  N_{v}(\bx)   \, \delta_{| u - v |, 1}$.

\subsection{The main result}
Our main result establishes the existence, for any dimension and any number of particles, of a regime $ \mathcal{L}^{(p)}_n $ in the space of basic parameters of the model, $(\lambda, \balpha) \in \mathbb{R}_+ \times \mathbb{R}^{p} $, for which
complete dynamical localization with a uniform localization length occurs for up to $ n $ particles. 
The proof gives an inductive algorithm for the  construction of  such  localization regimes, albeit at what may be a possibly far too restrictive manner.   

The first proof of spectral localization for two particles was presented in  \cite{CS09}, using the multiscale approach.  The result presented below was derived by different means, which allow also rather simple control of the dynamical localization.  

The localization regime covered by the result presented below may be best described in 
terms of the extreme sets it includes, namely:

\begin{description}
   \item[Strong disorder]  for  each $\balpha \in \R^{p}$ there is $\lambda(\balpha )$ such that the localization regime   $\mathcal{L}^{(p)}_n $ includes the cone in the parameter space $ \mathbb{R}_+ \times \mathbb{R}^{p} $ 
   where the interaction strengths are dominated by $\balpha$, and the disorder strength exceeds  $ \lambda(\balpha)$.

  \item[Weak interactions]   for any  $\lambda > \lambda_1 $, i.e. disorder strength at which the one-particle Hamiltonian exhibits complete localization,  
   there are  $\balpha_j(\lambda) > 0$
 , $j=\{1,...,p\}$, 
  such that $ \mathcal{L}^{(p)}_n$ includes all $(\lambda, \balpha') $ for which 
  $|\balpha_j'| \le  |\balpha_j (\lambda)|$ componentwise.  
    \end{description}
Thus,  
the weak localization region includes some neighborhood of the {\it entire}  localization regime $ (\lambda_1,\infty) \times \{ \mathbf{0} \} $ of the $n$ particle 
unperturbed system. In case $ d = 1 $, $\lambda_1 = 0$, so this includes all positive values of~$ \lambda $. 

Having introduced these notions we may present the result of \cite{AW09}: 
\begin{theorem}[\!\!\cite{AW09}]\label{eq:main}
For each $ n , p \in \N$ there is an open set $\mathcal{L}^{(p)}_n \subset \R_+\times \R^{p} $ 
which  includes regimes of strong disorder and weak interactions, for which 
at some $ A ,\xi <\infty $ and all $(\lambda, \balpha) \in\mathcal{L}^{(p)}_n$, $ k \in \{ 1, \dots , n \} $,   and all $\bx,\, \by \in (\Z^d)^k$:
\begin{equation} \label{loc_f}
\E{\left[\sup_{ t \in \mathbb{R}}\left|\langle \delta_\bx \, ,  \, e^{-itH^{(k)}} \, \delta  _\by \rangle \right|^2\right]}   \ \le \  A \, 
e^{-  \rm{dist}_{\mathcal H}(\bx,\by) / \xi}   \, ,
\end{equation}
where the exponential decay is in terms of the Hausdorff pseudo-distance between the configurations  $ \bx = (x_1, \dots , x_k) $ and $ \by = (y_1, \dots , y_k) $  
\begin{equation} \label{eq:distH}
\dist_{\mathcal H}(\bx, \by) 
\ := \ \max\left\{ \max_{1\le i \le k} \dist(x_i,\, \by), \   \max_{1\le i \le k} \dist(y_i,\,\bx)
\right\}    
\end{equation} 
\end{theorem}

\noindent As was mentioned, spectral  localization  follows from \eqref{loc_f}  by the Wiener criterion.   



\subsection{Comments on the result} \label{comments}

Theorem~\ref{eq:main} is formulated for distinguishable particles. For  {\it Fermions}, or {\it Bosons}, the relevant Hilbert spaces are subspaces of the space which this theorem covers, and thus the results apply by restriction.  Furthermore, {\it hard core interactions}  can also be added without  requiring  any modification in the proof of the  extended  statement.


It is worth pointing out that for systems of  $n>2$ particles some subtleties show up in the decay rate seen in~\eqref{loc_f}.  
A natural decay rate  for systems with permutation symmetry is the symmetrized distance: 
$
 \dist_S(\bx,\by) \ := \ \min_{\pi \in S_n}  \sum_{j=1}^n |x_j-y_{\pi j}| \,  ,  
$ 
where $S_n$ is the permutation group of 
$\{1, ..., n\}$.  In contrast, $ \dist_{\mathcal H}(\bx,\by)$ is not a metric, if  $n>2$. 
The Hausdorff pseudo-distance bounds~\eqref{loc_f} allow for the possibility that if the initial configuration had some  particles  within the localization distance from each other  then such `excess charge' may transfer among the occupied regions.
 \begin{figure}[h]
\begin{center}
\includegraphics[width=.4\textwidth] {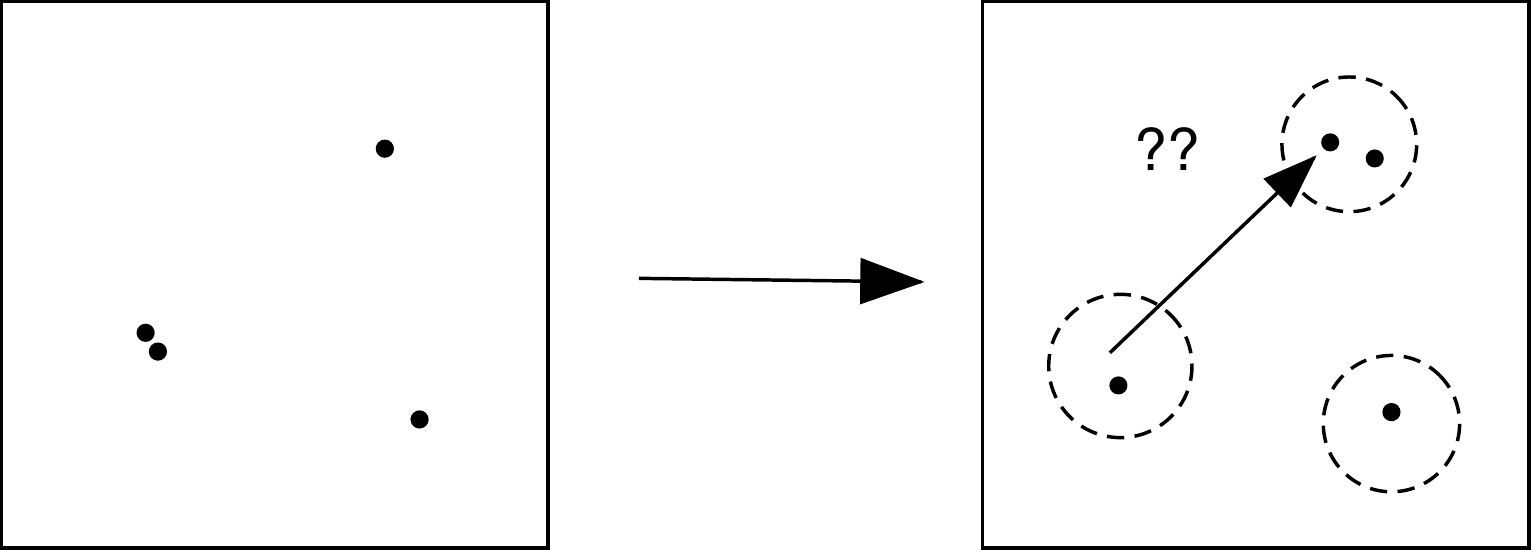}\\[1ex]
\caption{\small A transition  between two configurations  which are close in the Hausdorff pseudo-distance, but not otherwise.}
\end{center} 
\end{figure}

One could   wonder whether Theorem~\ref{eq:main} does not immediately follow from the existing localization results.
After all,  a configuration of $ n $ particles in $ \Z^d $ may be regarded as a single particle in $ \mathbb{Z}^{dn}$   evolving  under a  Hamiltonian which is  of the form
$
H^{(n)}(\omega) = - \Delta^{(nd)} + U(\bx) + \lambda \, V(\bx;\omega) 
$
with a random potential $ V(\bx;\omega) = \sum_{j=1}^n V(x_j;\omega)  $.   
However, the random potential $V(\bx;\omega) $  does not take independent values at configurations with one coinciding particle and in this sense is of \emph{infinite range}.  
The number of its degrees of freedom is only a fractional power~($\frac{1}{n}$) of the number of configuration.   
The technical problem which this causes  is behind the above mentioned limitation of the bound~\eqref{loc_f}.   

\section{Outline of the proof}

The proof of Theorem~\ref{eq:main} cannot be fitted in this short summary, but we may comment on  the flow of the argument \cite{AW09}.  It  consists of three steps:
\begin{enumerate}
\item  Proof of the finiteness of fractional moments of the Green function. 
\item  Elucidation of the relation of a multiparticle  \emph{eigenfunction correlator} with the Green function's fractional moments.
\item Inductive construction of domains of uniform localization in parameter space.
\end{enumerate}
The first is a  \emph{sine qua non} condition for the  
fractional moment technique for localization~\cite{AM}.  
It is related to the celebrated Wegner estimate, which however  does not play a direct role in the analysis.    The proof is rather easy, as is also the case  for the   Wegner estimate for many-particle systems.

The second step follows the strategy by which 
dynamical localization is established for single particles \cite{A2}, for which some adjustments to the multi-particle setup are required.
 
Once the first two preparatory steps are taken, the third step  constitutes the core of the argument.   



\subsection{Finiteness of fractional moments of the Green function }

An essential tool for the analysis is the finite-volume, $ \Lambda \subset \mathbb{Z}^d $, Green function 
\begin{equation}
G^{(n)}_\Lambda(\bx,\by;z) := \big\langle \delta_\bx , \big( H^{(n)}_\Lambda - z\big)^{-1} \delta_\by \big\rangle \, ,
\end{equation}
with $  H^{(n)}_\Lambda $ denoting the restriction of \eqref{eq:defHm} to $ \ell^2(\Lambda)^n $. 
The first step towards \emph{exponential bounds} is to prove finiteness  of the fractional moments.
In fact, even a conditional average makes them finite. 
\begin{lemma}\label{prop:Wegner}
For any $ s \in (0,1) $ there is $ C  < \infty $ such that for   
any $ \Lambda  \subseteq \mathbb{Z}^{d}$, any two (not necessarily distinct) sites $ u_1,u_2 \in \Lambda $ and any pair  of configurations
$ \bx ,\by$,  which have a  particle at $ u_1 $, $ u_2 $ respectively,
\begin{equation}\label{eq:Wegner}
	\mathbb{E}\left( \left| G^{(n)}_\Lambda(\bx,\by;z)\right|^s \, \Big| \, \left\{ V(v) \right\}_{v\not\in \{u_1,u_2\}} \right) \leq \frac{C}{|\lambda|^s}  
\end{equation}
for all $ z \in \mathbb{C} $, and $ \lambda \neq 0 $. 
\end{lemma}
\begin{proof}[Sketch of proof]
 In its dependence on $V(u_1)$ and $V(u_2)$,  the Hamiltonian is of the form: 
$ H^{(n)}_\Lambda = A + \lambda \, V(u_1)  \, N_{u_1} + \lambda \, V(u_2)  \, N_{u_2} $, 
where $ \left(N_u \psi \right)(\bx) := \sum_{k=1}^n \delta_{x_k,u} \, \psi(\bx)  $ stands for the number operator.
The assertion is therefore implied by the weak-$L^1$ estimate from \cite{AENSS}: 
\begin{equation*}
	\int_{\mathbb{R}^2} \!\! \indfct \left[  \left| \left\langle \sqrt{N}\phi , \left( \xi \, N+ \eta \, M   -  K \right)^{-1} \sqrt{M}\psi \right\rangle \right| > t \right]  \varrho(\xi) \, \varrho(\eta) \, d\xi \, d\eta 
		\leq \frac{C[\varrho]}{t} \; \| \phi \| \, \|\psi \| \, , 
\end{equation*}
with operators $ N, M \geq 0 $ and $ K $ dissipative.
\end{proof}
The finiteness~\eqref{eq:Wegner} is the analogue of a {\it Wegner estimate} in the multi-scale method. As has been noticed in~\cite{ChSu2,Ki08}, 
 it
follows similar to the one-particle case from the monotonicity used in the above proof.

\subsection{Eigenfunction correlator and its relation to Green function}

For the finite-volume Hamiltonian $ H^{(n)}_{\Lambda} $, we define the \emph{eigenfunction correlator}, associated with
an energy regime $ I \subset \mathbb{R} $, as the sum 
\begin{equation}
	Q^{(n)}_{\Lambda}(\bx,\by;I) := \mkern-20mu 
	\sum_{E  \in \sigma(H^{(n)}_{\Lambda})\cap  I }  \mkern-15mu 	 \left| \langle \delta_\bx \, ,  P_{\{E\} }(H^{(n)}_\Lambda) \, \delta_\by \rangle \right|  \, , 
\end{equation}
where $P_{\{E\} }$  is a spectral projection. 
Ignoring some subtleties related to the passage to the infinite-volume limit $ \Lambda \uparrow \mathbb{Z}^d $ (which are discussed in detail in~\cite{A2,AW09}),  
the elementary bound $\left| \langle \delta_\bx, e^{-it H^{(n)}_{\Lambda})} \delta_\by \rangle \right| \leq Q^{(n)}_{\Lambda}(\bx,\by;\mathbb{R}) $
shows that the eigenfunction correlator is an important tool for establishing dynamical localization \cite{KuSou80,A2}. \\

The following key lemma states the equivalence of exponential decay in the Hausdorff pseudo-distance (though not in the symmetric distance) of the eigenfunction correlator 
and of the Green function's fractional moments.\\  
\begin{lemma}
The following statements are equivalent:
\begin{enumerate}
	\item There is $ A, \xi < \infty $ such that for all $ \bx , \by \in (\Z^d)^k $: 
	\begin{equation}
\sup_{\substack{I \subset \R }} \sup_{\Lambda \subset \mathbb{Z}^d}  \E{\left[Q_{\Lambda}^{(k)}(\bx,\by;I)\right]}   \ \le \  A \; 
e^{-  \dist_\mathcal{H}(\bx,\by) /\xi  } \,  .
\end{equation}
	\item There is $ A, \xi < \infty $ and $ s \in (0,1 ) $ such that for all $ \bx , \by \in (\Z^d)^k $: 
	\begin{equation}
	 \sup_{\substack{I \subset \R \\ |I| \geq 1}} \sup_{\Lambda \subset \mathbb{Z}^d}\frac{1}{ |I|}  \underbrace{\int_I \mathbb{E}\left[ \left| G^{(k)}_\Lambda(\bx,\by;E)\right|^s    \right] dE}_{=: \, \widehat{\mathbb{E}}_I\left[ \dots \right]}   \ \leq \ 
	A\,  e^{-\dist_\mathcal{H}(\bx,\by)/\xi} \, . 
	\end{equation}
\end{enumerate}
\end{lemma}

In the proof of the main result we repeatedly change horses between estimates on the kernels of $G$ and  $Q$.  The kernel $Q(\bx,\by) $ is what is ultimately need, and  
it is also a very convenient tool for the induction step.  
Yet, $G(\bx,\by)$ has a more convenient perturbation theory.



In the induction step described below one puts together two mutually non-interacting subsystems and then turns on the interaction between them.   Thus, we start with the system partitioned into 
$ J, K \subset \{ 1 , \dots , n \} $, $ J \cap K = \emptyset $, and   a  Hamiltonian of the form   
$ H^{(J,K)}_\Lambda := H^{(J)}_\Lambda \oplus H^{(K)}_\Lambda $ acting on  $ \ell^2(\Lambda)^{|J|} \otimes \ell^2(\Lambda)^{|K|} $.  Then: 

\begin{itemize}
\item The Green function of the composite system  is a 
 convolution, 
 \begin{equation}
 	 G^{(J,K)}_\Lambda(\bx,\by;z) = 
 \int_\mathcal{C} G^{(J)}_\Lambda(\bx_J,\by_J;z-E) \, G^{(K)}_\Lambda(\bx_K,\by_K;E) \; \frac{dE}{2\pi i }  
\end{equation}
involving a contour integral encircling the spectrum of $ H^{(K)}_\Lambda $.  Moment estimates for the combined system are  complicated by the unboundedness of $G$ and possible correlations between the  eigenvalues of the subsystems. 
\item On the other hand,  the eigenfunction correlator is bounded and  we may use:
\begin{equation}\label{eq:ecpartial}
	Q^{(J,K)}_\Lambda(\bx,\by;I) 
		\leq \,  Q^{(J)}_\Lambda(\bx_J,\by_J; \mathbb{R}) \; Q^{(K)}_\Lambda(\bx_K,\by_K; \mathbb{R}) \, ,  
\end{equation}
by which the exponential decay of the subsystems  passes to the joint system. 
\end{itemize}
\subsection{Inductive construction of domains of  localization in parameter space}
In view of the equivalence of the exponential decay in the Hausdorff distance of the eigenfunction correlator 
and of the fractional moments of the Green function, it is natural to define localization regimes as  follows. 
\begin{definition} \label{def:loc}
An open subset of the parameter space, $\mathcal{L} \subset \R_+\times \R^{p}$ is said to be a {\it domain of uniform $n$-particle localization}
if for some $s\in (0,1)$ there exists $\xi  < \infty$ and $ A  < \infty $ such that 
$$
(\mbox{\rm UL}) \qquad \qquad \sup_{\substack{I \subset \R \\ |I| \geq 1}} \sup_{\Lambda \subset \mathbb{Z}^d} \frac{1}{ |I|}  \int_I \mathbb{E}\left[ \left| G^{(k)}_\Lambda(\bx,\by;E)\right|^s    \right] dE \ \leq \ 
	A\,  e^{-\dist_\mathcal{H}(\bx,\by)/\xi} \, . 
	$$
 holds for all $ (\lambda,\balpha) \in \mathcal{L} $, all $k\in \{1,...,n\}$, and all  $\bx, \by \in (\Z^d)^k$. 
\end{definition}
Our aim is to inductively construct $ \mathcal{L}_{n}^{(p)} $ starting from the domain $  \mathcal{L}_{1}^{(p)} = (\lambda_1,\infty) \subset \mathbb{R}_+ $ of uniform one-particle localization,  whose existence is guaranteed in \cite{A2} (and \cite{KuSou80} for $ d = 1 $).
For the induction step we assume that $ \mathcal{L}_{n-1}^{(p)} $ is a domain of uniform $ (n-1) $-particle localization and proceed by distinguishing 
two cases:
\begin{enumerate}
\item{\bf Localization for non-clustered configurations}.
Here we deal with two configuration $ \bx, \by $ of which at least one  is of   diameter ($ \diam(\bx) \ := \ \max_{j,k}\,   |x_j - x_k|  $) comparable with their Haussdorff distance.  A key lemma provides the bound:
\begin{multline}\label{eq:boundcluster}
	\qquad \quad  \sup_{\substack{I \subset \mathbb{R} \\ |I| \geq 1}} \sup_{\Lambda \subseteq \mathbb{Z}^d }\, 
	\widehat{\mathbb{E}}_I\left[ |G_\Lambda^{(n)}(\bx,\by)|^{s}\right] \\
	 	\ \leq \  A  \, \exp\left( - \frac{1}{\xi} \min\left\{ 
		 	\dist_{\mathcal H}(\bx,\by), \frac{\max\{\diam(\bx), \diam(\by)\}}{n-1}\right\} \right)  
\end{multline}
for all $ (\lambda, \balpha ) \in \mathcal{L}_{n-1}^{(p)} $, (which, strictly speaking, is only valid under the additional assumption of $\mathcal{L}_{n-1}^{(p)}  $ being sub-conical, cf.~\cite{AW09}).  
\end{enumerate}
\begin{figure}[h]
\begin{center}
\begin{minipage}{.45\textwidth}
 \includegraphics[height=.2\textheight,width=.8\textwidth]{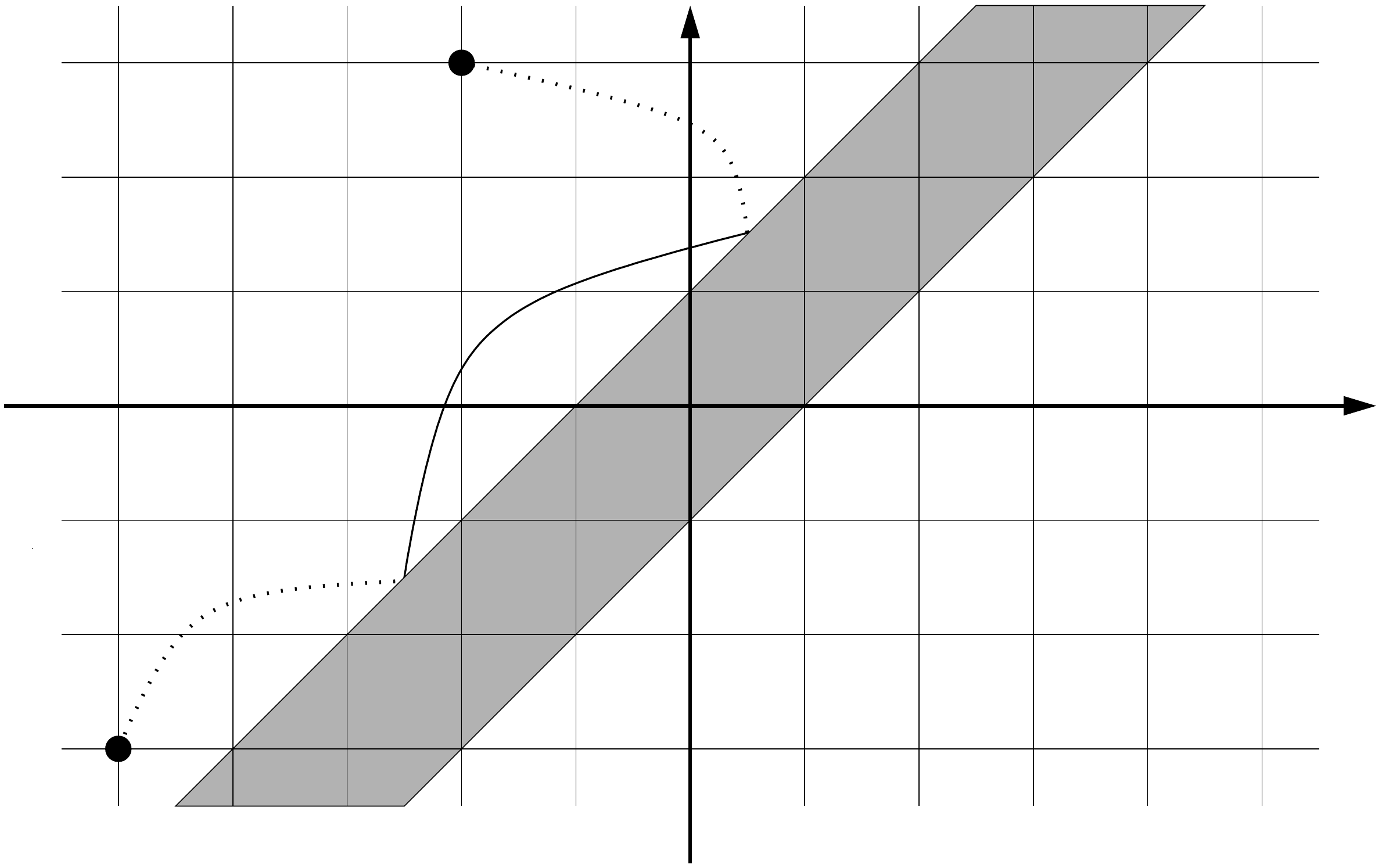}
\end{minipage}\hfill\begin{minipage}{.5\textwidth}
\caption{\small The solid line is the full  and
the dotted lines represent the Green function of the partially non-interacting system which, by assumption and due to arguments based on~\eqref{eq:ecpartial} decay exponentially. Due to finite range, the interaction is localized  to the shaded strip of width $2 \ell_U $ about the diagonal.}\label{fig2}
\end{minipage}
\end{center}
\end{figure}
The proof of~\eqref{eq:boundcluster} proceeds by breaking the configuration in two parts and removing the \emph{inter}cluster interaction. 
Standard perturbation theory then yields an expansion in $ U $ which in case  $ n = 2 $, $ d = 1 $ is pictorially explained in {\it configuration space} in Fig~\ref{fig2}.

The bound~\eqref{eq:boundcluster} is rather crude in its dependence on the particle number. It is the main reason why our bounds on the 
localzation length degrade rapidly with~$ n $.

\begin{enumerate}
\item[(ii)] {\bf Localization for clustered configurations}.  This refers to configurations with diameter less than half their separation.    The issue which is to be addressed is the possible formation of a quasi-particle which is not constrained by previous localization bounds.

   A convenient quantity to monitor is 
\begin{equation}
 B^{(n)}_s(L) := \sup_{\substack{ I \subset \mathbb{R} \\ |I | \geq 1 }}
  \sup_{\Omega \subseteq \Lambda_L} \,  |\partial \Lambda_L|\,  \sum_{y\in \partial \Lambda_L} \sum_{\substack{\bx \in\mathcal{C}^{(n)}_{L}(\Omega;0)\\ \by \in\mathcal{C}^{(n)}_{L}(\Omega;y)} } 
 	\widehat{\mathbb{E}}_I\left[ \left|  G^{(n)}_\Omega(\bx,\by) \right|^{s} \right] \, , 
\end{equation}
where  $ \mathcal{C}^{(n)}_{L}(\Lambda;x)\; $ denotes the collection of configurations of diameter less than $ L/2  $ and at least one particle at $ x $. \\ 

\noindent
\begin{minipage}{.3\textwidth}
 \includegraphics[height=.13\textheight,width=.8\textwidth]{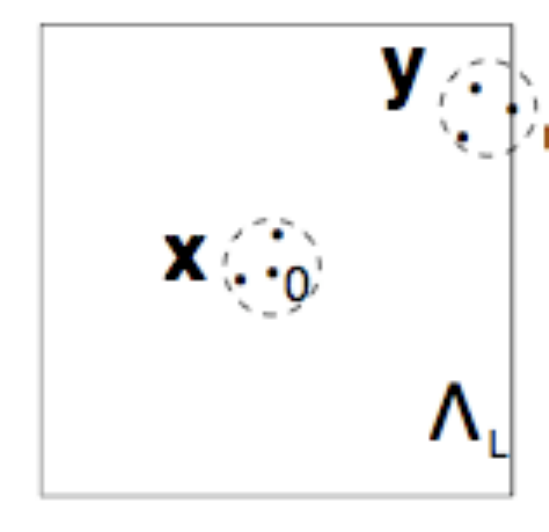}
\end{minipage}\hspace*{-3ex}
\begin{minipage}{.6\textwidth}
\begin{itemize}
\item one has at least one particle at $ 0 $, and
\item the other one has at least one particle at the boundary $ \partial \Lambda_L $ of the box; cf. illustration in {\it real space}.
\end{itemize}
\end{minipage}\\
\end{enumerate}

Localization bounds for clustered configurations proceed through {\it rescaling inequalities} for  $  B^{(n)}_s(L)  $.    The key statement is:

\begin{lemma}
 There exists $ s \in (0,1) $, $ a, A  , p < \infty $, and $ \nu > 0 $ such that 
\begin{equation}\label{eq:rescaling}
	 B^{(n)}_s(2L ) \leq \frac{a}{|\lambda|^{s} }\, B^{(n)}_s(L)^2 + A  \, L^{2p} \, e^{- 2\nu L} \
\end{equation}
for all $(\lambda,\balpha) \in \mathcal{L}^{(p)}_{n-1} $,  (which is again assumed to be sub-conical, cf.~\cite{AW09}). 
\end{lemma}

Part of the proof of this lemma resembles the proof of finite-volume criterial for one particle localization   in \cite{ASFH}. 
In view of the infinite-range correlations 
of the random potential in configuration space discussed at the beginning of this section, a crucial observation is the following. 
When placing two 
boxes of length $ L/2 $ about configurations $ \bx $ and $ \by $ in the above picture, the random variables associated with those boxes are independent 
due to the spatial separation. 
Moreover, the error when restricting the sum in the definition of $ B^{(n)}_s(2L) $ to configurations $ \bx , \by $ which have a diameter less than $ L/4 $ can be controlled by~\eqref{eq:boundcluster} yielding the second term on the right side in~\eqref{eq:rescaling}.\\

It is not hard to see that rescaling inequalities such as~\eqref{eq:rescaling} imply exponential decay provided the quantity $ \frac{a}{|\lambda|^{s} }\, B^{(n)}_s(L) $ is small on some scale $ L$. 
This is the requirement which determines $\mathcal{L}^{(p)}_{n} $. Namely, 
\begin{itemize}
\item in case of
 {\it strong disorder} and arbitrary value of the interaction $ \balpha $ we choose $L$ fixed and $ \lambda $ small enough. 
 \item in case of {\it weak interaction} and arbitrary value of the disorder $ \lambda > \lambda_1 $ we appeal to \cite{AM,ASFH} which imply that
  $ B^{(n)}_s(L) \to 0 $ as $ L \to \infty $ for $ \balpha = \mathbf{0} $. Since the finite-volume quantity $ B^{(n)}_s(L) $ is continuous in some neighborhood of $ \balpha = \mathbf{0} $ we choose $ L $ and $ \balpha $ accordingly. 
 \end{itemize}

\section{Some remaining challenges}
 
While Theorem~\ref{eq:main} is formulated for interactions of finite range, its proof allows for extension to interactions of with exponential falloff.  However, it it does not address  questions about the effects of Coulomb interactions.

An important outstanding challenge it to resolve the question which is commented upon in Section~\ref{comments}: does  localization persists (with uniform bounds)  even for  large systems at  a positive density of particles?     
%

\section*{Acknowledgments}
This work was partially supported by the National Science Foundation under grants DMS-0602360 (MA), DMS-0701181 (SW) and a Sloan Fellowship~(SW).  

\bibliographystyle{ws-procs975x65}
\bibliography{ws-pro-sample}

\end{document}